\documentclass[11pt,a4paper]{article}
\usepackage{amsfonts,amssymb,graphics,epsfig,verbatim,bm,latexsym,amsmath,url,amsbsy,geometry,subdepth,amsthm,mathpazo,floatrow,xcolor,setspace,float}

\usepackage[utf8]{inputenc} % Required for inputting international characters
\usepackage[T1]{fontenc} % Output font encoding for international characters

\theoremstyle{plain}
\newtheorem{assumption}{2SLS}

\newtheorem{proposition}{Proposition}

\def\AmSTeX{$\cal A$\kern-.1667em\lower.5ex\hbox{$\cal M$}\kern-.125em
            $\cal S$-\TeX}

\usepackage{subcaption}
\usepackage[round]{natbib}

\begin{document}

\title{On the Effect of Imputation on the 2SLS Variance}
\author{Helmut Farbmacher$^{a,b,}$\thanks{%
farbmacher@econ.lmu.de}\\ $^{a}${\small Department of Economics,} {\small University of Munich, Germany}
\\ $^{b}${\small Munich Center for the Economics of Aging,} {\small Max Planck Society, Germany} \and Alexander Kann$^{c}$\\$^{c}${\small Department of Economics, University of Mannheim, Germany}}
%EndAName
\date{\today}
\maketitle

\begin{abstract}
%TCIMACRO{\TeXButton{noindent}{\noindent} }%
%BeginExpansion
\noindent
%EndExpansion
Endogeneity and missing data are common issues in empirical research. We investigate how both jointly affect inference on causal parameters. Conventional methods to estimate the variance, which treat the imputed data as if it was observed in the first place, are not reliable. We derive the asymptotic variance and propose a heteroskedasticity robust variance estimator for two-stage least squares which accounts for the imputation. Monte Carlo simulations support our theoretical findings.

\medskip

%TCIMACRO{\TeXButton{noindent}{\noindent}}%
%BeginExpansion
\noindent%
%EndExpansion
\end{abstract}
\onehalfspacing
\baselineskip=19pt

%----------------------------------------------------------------------------------------
%	Main Part
%----------------------------------------------------------------------------------------

%	Introduction
%----------------------------------------------------------------------------------------

\section{Introduction}

Missing data are frequently encountered in empirical studies in economics and social sciences. A popular method to handle missing data is the complete case approach, which excludes incomplete observations from the analysis. Among others, an alternative approach is regression imputation, which utilizes the complete observations to fill in the missing values. The imputed data is then used as if it was observed in the first place. 

While the imputation of an exogenous regressor is a well-researched topic \citep[e.g.,][]{little1992missingx}, little is known about the imputation of an endogenous regressor. Two-stage least squares (2SLS) estimation is a way to deal with the endogeneity of a regressor. \cite{MCDONOUGH2017} discussed the bias of 2SLS with imputation. Their analysis was mainly based on \cite{NAGAR1959}'s seminal work about the finite-sample bias of 2SLS. However, they did not discuss variance estimation, which is challenging as we show in this study. 2SLS inference is affected by the imputation implying that the conventional variance estimator, which ignores the imputation, is only valid if we are interested to test whether the parameter of the endogenous regressors is zero. Standard errors and confidence intervals based on the conventional variance are invalid. We obtain the asymptotic distribution and derive a heteroskedasticity robust variance estimator for 2SLS with regression imputation, which allows to construct valid standard errors, confidence intervals and conduct tests.

We focus on settings where the endogenous regressor is missing at random, the number of instruments is fixed as the sample size grows and the imputation method is regression imputation. The missing at random setting is essentially selction based on observables \citep{wooldridge2007} and covers many relevant applications \citep[for instance,][for applications with missing exogenous variables]{wooldridge2007,graham2012,chaudhuri2018}. To simplify the comparison with the conventional variance estimation, we assume missing completely at random in parts of the study.

We illustrate our theoretical results using Monte Carlo simulations. In the simulation results we focus on the key parameters in the interplay between 2SLS estimation and regression imputation. These key parameters are the fraction of missings, which can be observed easily, and the direction of the OLS bias, which is unobserved but very often researchers have strong prior beliefs about it. 

In the next section we describe the setup including model, missing data structure and a brief description of regression imputation. In Section \ref{chap:inference}, we derive the asymptotic variance and propose a variance estimator which accounts for the imputation and is robust to general forms of heteroskedasticity. Section \ref{sec:MCsim} shows Monte Carlo simulation results and Section \ref{sec:conclusion} concludes.

%	Setup
%----------------------------------------------------------------------------------------

\section{Setup} \label{sec:setup}

\subsection{Model} \label{sec:2sls}

Consider the standard simultaneous equation model
\begin{align} \label{eq:molde}
	\begin{split}
		y_i&=x_i\beta+u_i\\
		x_i&=Z_{i.}'\pi+v_i \, ,
	\end{split}
\end{align}
with dependent variable $(y_i)$, endogenous regressor $(x_i)$ and a set of $L$ instruments $(Z_{i.})$.\footnote{To simplify notation, we abstract from exogenous control variables.} The parameter of interest is $\beta$. The two-stage least squares (2SLS) estimator is a way to deal with the endogeneity of $x$:
\begin{align} 
	\begin{split}
		\widehat{\beta}&=(x'P_Zx)^{-1}x'P_Zy \, ,
	\end{split}					
\end{align}
where $P_Z=Z(Z'Z)^{-1}Z'$. Throughout the paper we assume that 2SLS\ref{ass:2sls1}-2SLS\ref{ass:2sls4} are fulfilled in model (\ref{eq:molde}), which assures consistency and valid inference of the 2SLS estimator. 
\begin{assumption} \label{ass:2sls1}
	$plim\left(\frac{Z'u}{n}\right)=E[Z_{i.}u_i]=0 \, $; $plim\left(\frac{Z'v}{n}\right)=E[Z_{i.}v_i]=0 \,$.
\end{assumption}
\begin{assumption} \label{ass:2sls2}
	$plim\left(\frac{Z'x}{n}\right)=E[Z_{i.}x_i]=Q_{Zx}\neq0 \, $.
\end{assumption}
\begin{assumption} \label{ass:2sls3}
	$plim\left(\frac{Z'Z}{n}\right)=E[Z_{i.}Z_{i.}']=Q_{ZZ}$, with $Q_{ZZ}$ a finite and full rank matrix.
\end{assumption}
\begin{assumption} \label{ass:2sls4}
	Observations are identically and independently distributed.
\end{assumption}
\begin{assumption} \label{ass:2sls5}
	Errors are homoskedastatic $E[uu'|Z]=\sigma_{u}^2I_n \,$; $E[vv'|Z]=\sigma_{v}^2I_n \,$;  $E[vu'|Z]=\sigma_{uv}I_n \, $; $\sigma_{u}^2$, $\sigma_{v}^2$ and $\sigma_{uv}$ are finite.
\end{assumption}

\subsection{Missing data} \label{sec:missing_data}
The analysis is, however, complicated by the fact that data on the endogenous regressor, $x$, is missing for some observations causing them to be incomplete. Therefore, $\widehat\beta$ as defined above cannot be calculated with the data at hand. Let the subscripts indicate the missing status, that is, $y_0, x_0, Z_0$ indicate the complete observations and $y_1, x_1, Z_1$ the observations with missing value of $x$. Let $\widehat{p}$ be the probability of missings in the endogenous regressor and assume
\begin{assumption} \label{ass:2sls5b}
	$\widehat{p}=\frac{n_1}{n}\rightarrow{}p<1$ as $n\to\infty \, $,
\end{assumption}
\noindent
which implies that not only the number of incomplete observations ($n_1$) but also the number of complete observations ($n_0$) increase as $n$ increases. Without loss of generality we assume that the first $n_0$ observations of the matrix $Z$ and the vectors $x$ and $y$ are complete observation and the remaining $n_1$ observations are incomplete. 

The missing data literature distinguishes three types of missing structures: missing completely at random (MCAR), missing at random (MAR)\footnote{That is the missing structure only depends on observables.}, and not missing at random (NMAR). While MCAR is the easiest to deal with, real data is often NMAR or MAR. The missing structure is called ignorable if the data is either MCAR or MAR. Since dealing with NMAR is very different from the other two types, we focus on missing (completely) at random in this exposition, and use the following additional assumption:

\begin{assumption} \label{ass:2sls6}
	The missing structure is ignorable. 
\end{assumption}

\subsection{Regression imputation} \label{sec:RI}
Regression imputation (RI), which is a two-step procedure, can be applied as a tool to fill in the missing values. In the first step the complete observations are used to regress the endogenous variable $(x_0)$ on the imputation variables to obtain the imputation parameters. These parameters are equal to the first stage estimates from the complete case approach $\left(\widehat\pi_{CC}=(Z_0'Z_0)^{-1}Z_0'x_0\right)$ if the imputation method incorporates the instruments.\footnote{\cite{MCDONOUGH2017} show in Monte Carlo simulations that imputation methods, which incorporate the instruments, produce the smallest finite sample bias of 2SLS estimation. } In the second step these estimates are employed to impute the missing values in $x_1$ by multiplying the imputation variables for the incomplete observations with the parameters obtained from the first step. The imputed variable is then used in the 2SLS estimation as if it was observed in the first place.\footnote{Clearly, the observations with missing values cannot add any information to the estimation of the first stage parameters. Hence, the relevant $F$-statistic for the first stage parameters should be based on the complete case observations.} 

\noindent
The model in (\ref{eq:molde}) can be restated in terms of $\widetilde{x}$ by adding an imputation error:
\begin{align} \label{eq:imputed_model_general}
	\begin{split}
		y&=\widetilde{x}\beta+u-e\beta=\widetilde{x}\beta+\widetilde{u} \\
		\widetilde{x}&=Z\pi+v+e=Z\pi+\widetilde{v}  \\
		\widetilde{x}&=x+e=\begin{pmatrix}x_0\\x_1\end{pmatrix}+\begin{pmatrix}0\\e_1\end{pmatrix}=\begin{pmatrix}x_0\\\widetilde{x}_1\end{pmatrix} \, .
	\end{split} 
\end{align}
The imputation error $e_1$ and the composite errors of the imputed model ($\widetilde{v}$ \& $\widetilde{u}$) are defined as
\begin{align} \label{eq:imp_errors}
	\begin{split}
		e_1&=Z_1(\widehat{\pi}_{CC}-\pi)-v_1=Z_1(Z_0'Z_0)^{-1}Z_0'v_0-v_1\\
		\widetilde{v}&=\begin{pmatrix}v_0\\\widetilde{v}_1\end{pmatrix}=\begin{pmatrix}v_0\\v_1+e_1\end{pmatrix}=\begin{pmatrix}v_0\\Z_1(Z_0'Z_0)^{-1}Z_0'v_0\end{pmatrix}\\
		\widetilde{u}&=\begin{pmatrix}u_0\\\widetilde{u}_1\end{pmatrix}=\begin{pmatrix}u_0\\u_1-e_1\beta\end{pmatrix}=\begin{pmatrix}u_0\\u_1+v_1\beta-Z_1(Z_0'Z_0)^{-1}Z_0'v_0\beta\end{pmatrix} \, .
	\end{split}
\end{align}

\noindent
While the 2SLS estimator with regression imputation remains consistent if 2SLS \ref{ass:2sls1} to \ref{ass:2sls3} and \ref{ass:2sls5b} to \ref{ass:2sls6} are fulfilled, inference is affected by the imputation. The independence across observations is violated in the imputed model as the complete data has been used to impute incomplete observations. In the next section, we derive a variance estimator which takes this into account. 

%	Inference
%----------------------------------------------------------------------------------------

\section{Estimation and Inference}  \label{chap:inference} 

We consider
\begin{align}
	\widehat\beta_{RI}=(\widetilde{x}'P_Z\widetilde{x})^{-1}\widetilde{x}'P_Z y \, ,
\end{align}
and derive its limiting distribution under heteroskedasticity in the next proposition.

\begin{proposition} \label{prop:asym_var_het}
	Under Assumptions 2SLS \ref{ass:2sls1}-\ref{ass:2sls4}, 2SLS \ref{ass:2sls5b}-\ref{ass:2sls6}, $E\left[y_i^4\right]$, $E\left[\widetilde{x}_i^4\right]$, $E\left[\Vert Z_{i.}\Vert^4\right]$ being finite and $n\rightarrow\infty$, we have that 
	\begin{align}
		\begin{split}
			\sqrt{n}(\widehat{\beta}_{RI}-\beta)\xrightarrow[]{d}N(0,V_{\widehat\beta_{RI}})
		\end{split}
	\end{align}
	with the asymptotic variance given by, 
	\begin{align*}
		V_{\widehat\beta_{RI}}=\left(Q_{xZ}Q_{ZZ}^{-1}Q_{Zx}\right)^{-1}Q_{xZ}Q_{ZZ}^{-1}\,\Omega{}\,Q_{ZZ}^{-1}Q_{Zx}\left(Q_{xZ}Q_{ZZ}^{-1}Q_{Zx}\right)^{-1}
	\end{align*}
	where,
	\begin{align*}
		\Omega=&E[u_{i}^2 Z_{i.}Z_{i.}']-2pE[u_{0i}v_{0i}Z_{0i.}Z_{0i.}']Q_{Z_0Z_0}^{-1}Q_{Z_1Z_1}\beta\\
		&+\frac{p^2}{1-p}Q_{Z_1Z_1}Q_{Z_0Z_0}^{-1}E[v_{0i}^{2}Z_{0i.}Z_{0i.}']Q_{Z_0Z_0}^{-1}Q_{Z_1Z_1}\beta^2\\
		&+2pE[u_{1i}v_{1i}Z_{1i.}Z_{1i.}']\beta+pE[v_{1i}^2 Z_{1i.}Z_{1i.}']\beta^2
	\end{align*}
\end{proposition}
\begin{proof}
	See Appendix \ref{app:asym_var_het}. 
\end{proof}
The first term of $\Omega$ corresponds to the standard asymptotic variance of 2SLS without missing data. The remaining terms can be attributed to regression imputation. If no data is missing ($p=0$) or $\beta=0$, $\Omega$ collapses to the standard asymptotic variance of 2SLS. The latter simplification (i.e., $\beta=0$) is a well-known result in the literature about generated regressors \citep[see, for example,][]{murphy1985estimation}. It allows valid inference based on the conventional variance estimator if and only if we are interested in tests of $\beta=0$, which often is of major interest in applications. However, it does not justify the construction of standard errors or confidence intervals. 

The estimation of $\Omega$ is challenging as we cannot obtain reliable residual estimates for the imputed observations, $u_{1}$ and $v_{1}$ . Hence, the first, fourth and fifth term of $\Omega$ cannot be simply estimated by using the corresponding residuals. However, it is possible to circumvent this issue by using the error of the imputed model $(\widetilde{u}_{1})$ for which residuals $(\widehat{\widetilde{u}}_{1})$ can be obtained. The variance estimator given in the following proposition is consistent under general forms of heteroskedasticity. 
\begin{proposition} \label{prop:consistent_rob_var}
	Under Assumptions 2SLS \ref{ass:2sls1}-\ref{ass:2sls4}, 2SLS \ref{ass:2sls5b}-\ref{ass:2sls6} and $E\left[y_i^4\right]$, $E\left[\widetilde{x}_i^4\right]$, $E\left[\Vert Z_{i.}\Vert^4\right]$ being finite, we have that 
	\begin{align*}
		\widehat{V}_{\widehat{\beta}_{RI}}=\left(\widetilde{x}'P_Z\widetilde{x}\right)^{-1}\widetilde{x}'Z(Z'Z)^{-1}\widehat{W}_{RI}(Z'Z)^{-1}Z'\widetilde{x}\left(\widetilde{x}'P_Z\widetilde{x}\right)^{-1}
	\end{align*}
	is a consistent estimator of the asymptotic variance ($n\widehat{V}_{\widehat\beta_{RI}}\overset{p}{\to}V_{\widehat\beta_{RI}}$) with
	\begin{align*}
		\widehat{W}_{RI}=&\left(\sum_{i=1}^{n}\widehat{\widetilde{u}}_{i}^2Z_{i.}Z_{i.}'\right)-2\left(\sum_{i=1}^{n_0}\widehat{\widetilde{u}}_{i}\widehat{\widetilde{v}}_{i}Z_{i.}Z_{i.}'\right)\left(\sum_{i=1}^{n_0}Z_{i.}Z_{i.}'\right)^{-1}\left(\sum_{i=n_0+1}^{n}Z_{i.}Z_{i.}'\right)\widehat{\beta}_{RI}\\
		&+\Bigg[\left(\sum_{i=n_0+1}^{n}Z_{i.}Z_{i.}'\right)\left(\sum_{i=1}^{n_0}Z_{i.}Z_{i.}'\right)^{-1}\left(\sum_{i=1}^{n_0}\widehat{\widetilde{v}}_{i}^2Z_{i.}Z_{i.}'\right)\left(\sum_{i=1}^{n_0}Z_{i.}Z_{i.}'\right)^{-1}\left(\sum_{i=n_0+1}^{n}Z_{i.}Z_{i.}'\right)\\
		&-\sum_{i=n_0+1}^{n}Z_{i.}Z_{i.}'\left(\sum_{i=1}^{n_0}Z_{i.}Z_{i.}'\right)^{-1}
		\left(\sum_{i=1}^{n_0}\widehat{\widetilde{v}}_{i}^2Z_{i.}Z_{i.}'\right)\left(\sum_{i=1}^{n_0}Z_{i.}Z_{i.}'\right)^{-1}Z_{i.}Z_{i.}'\Bigg]\widehat{\beta}_{RI}^2
	\end{align*}
	where $\widehat{\widetilde{u}}_i=y_i-\widetilde{x}_i\widehat\beta_{RI}$ and $\widehat{\widetilde{v}}_i=\widetilde{x}_i-Z_{i.}'\widehat\pi_{CC}$	
\end{proposition}
\begin{proof}
	See Appendix \ref{app:consistent_rob_var}. 
\end{proof}

In the following we compare the conventional variance estimator, which ignores the imputation, with the true asymptotic variance. To simplify the comparison, we assume homoskedasticity and MCAR. The asymptotic variance is then stated in the following corollary.
\paragraph{Corollary 1.}
\textit{Under the conditions of Proposition \ref{prop:asym_var_het}, 2SLS \ref{ass:2sls5} and MCAR the asymptotic variance is given by}
\begin{align}
	\begin{split}
		V_{\widehat\beta_{RI}}^{hom}=&\left(Q_{xZ}Q_{ZZ}^{-1}Q_{Zx}\right)^{-1}\sigma_{u}^2+\left(Q_{xZ}Q_{ZZ}^{-1}Q_{Zx}\right)^{-1}\left( \frac{p}{1-p} \right)\sigma_{v}^2 \,\beta^2 \, .
	\end{split} \nonumber
\end{align}
\begin{proof}
	See Appendix \ref{app:hom_var}. 
\end{proof}

\noindent
The conventional variance estimator calculated by standard statistical software is given by
\begin{align}
	\begin{split} \label{naivevariance}
		\widehat{V}_{\widehat{\beta}_{RI}}^{conv}=\left(\widetilde{x}'P_Z\widetilde{x}\right)^{-1}\widehat{\sigma}_{\widetilde{u}}^2,\phantom{aa}with\;\widehat{\sigma}_{\widetilde{u}}^2=\frac{1}{n}\sum^{n}_{i=1}\widehat{\widetilde{u}}_i^2=\frac{1}{n}\sum^{n}_{i=1}(y_i-\widetilde{x}_i\widehat\beta_{RI})^2 \, .
	\end{split}
\end{align}
and its limit under homoskedasticity and MCAR is,
\begin{align} 
	\label{naive2SLS}
	\begin{split}
		n\widehat{V}_{\widehat{\beta}_{RI}}^{conv}\xrightarrow{p}&\left(Q_{xZ}Q_{ZZ}^{-1}Q_{Zx}\right)^{-1}\sigma_{u}^2+\left(Q_{xZ}Q_{ZZ}^{-1}Q_{Zx}\right)^{-1} p \left(2\sigma_{uv}\beta+\sigma_{v}^2\beta^2\right) \, .
	\end{split}
\end{align}
Comparing the limit of the conventional estimators with the true asymptotic variance in Corollary 1, we can see that the conventional estimator does not reflect the true variance of $\widehat\beta_{RI}$. For instance, while the asymptotic variance always increases with the missing probability, the conventional estimator could even decrease in it (if $-2\sigma_{uv}\beta>\sigma_{v}^2\beta^2$). In this case the limit of the conventional estimator is smaller than the asymptotic variance without any missing data problem, which is clearly counterintuitive. Moreover, the degree of endogeneity ($\sigma_{uv}$) erroneously affects the conventional limit while it has indeed no effect on the true asymptotic variance.

%	Simulations
%----------------------------------------------------------------------------------------

\section{Monte Carlo Simulation}\label{sec:MCsim}

We illustrate our theoretical findings in Monte Carlo simulations. To implement the missing data problem, we first mimic the model in (\ref{eq:molde}) using the data generating process described below, and then randomly delete the value $x_i$ with probability $p$. The missing structure is thus MCAR. For each of the $n$ observations $Z_{i.}$ and $v_i$ are drawn from
$$Z_{i.}\sim{}N\left(0_{L},\frac{1}{L}I_L\right) \ ; \hspace{0.5cm} v_i\sim N(0,1) \, , $$ 
where $L$ denotes the number of instruments. We follow \cite{hausman2012instrumental} and \cite{chao2014testing} and define the structural error as, 
$$u_i=\sigma_{uv}v_i+\sqrt{\frac{1-\sigma_{uv}^2}{\phi+(0.86^2)}}(\phi\epsilon_{1i}+0.86\epsilon_{2i}),\;\;\epsilon_{1i}\sim{}N(0,Z_{i.}'Z_{i.}),\;\epsilon_{2i}\sim{}N(0,0.86^2),\;\phi=5 \, ,$$
where $\phi$ defines the strength of the heteroskedasticity. We set $\beta=0.5$, $L=3$, $n=1000$, number of Monte Carlo repetitions $R=5000$, and the step size, which we use to alter the probability of missings in the different simulations, to $\Delta{}p=0.005$. As derived in the previous section, the sign of the endogeneity is crucial. Hence, we show results for $\sigma_{uv}=0.3$ and $\sigma_{uv}=-0.3$. We set $\pi=\sqrt{\frac{F\,L}{n}}$. Note that we divide by $n$ and not by $n_0$. Therefore, $F$ defines the first-stage $F$-statistic in the entire sample containing both complete and incomplete observations, and the first-stage $F$-statistic in the complete case sample gradually decreases with the missing probability. We set $F=100$. That is, the complete case $F$-statistic is around $100$ at $p=1$ and around $20$ at $p=0.8$.

\begin{figure}[!h]	
	\includegraphics[width=\columnwidth,trim={2cm 10.5cm 2cm 10.5cm}]{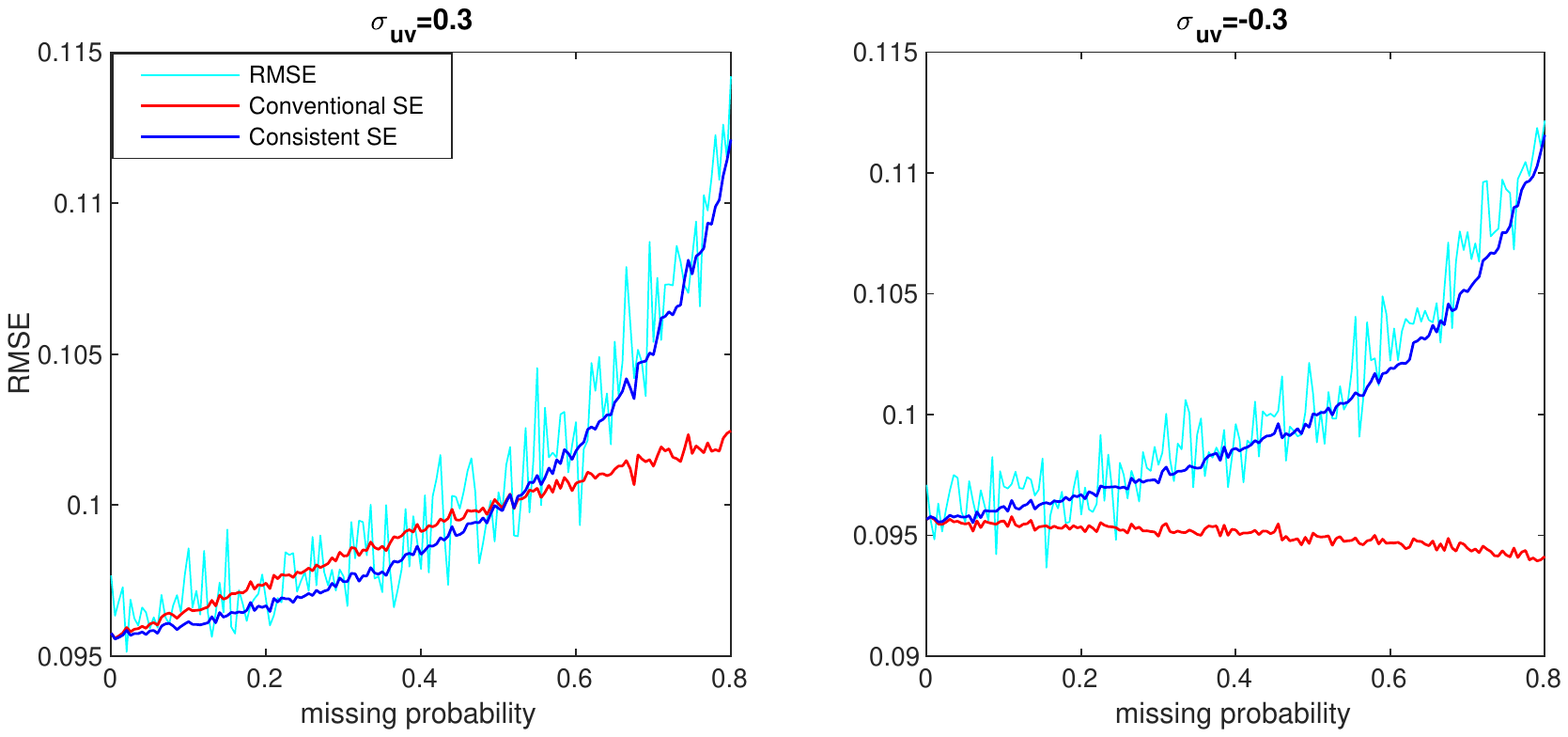}
	\caption{Consistent vs conventional variance estimator (RMSE)} \label{fig:cons_v_naive_rmse}
\end{figure}

\begin{figure}[h!]
	\includegraphics[width=\columnwidth,trim={2cm 10.5cm 2cm 10.5cm}]{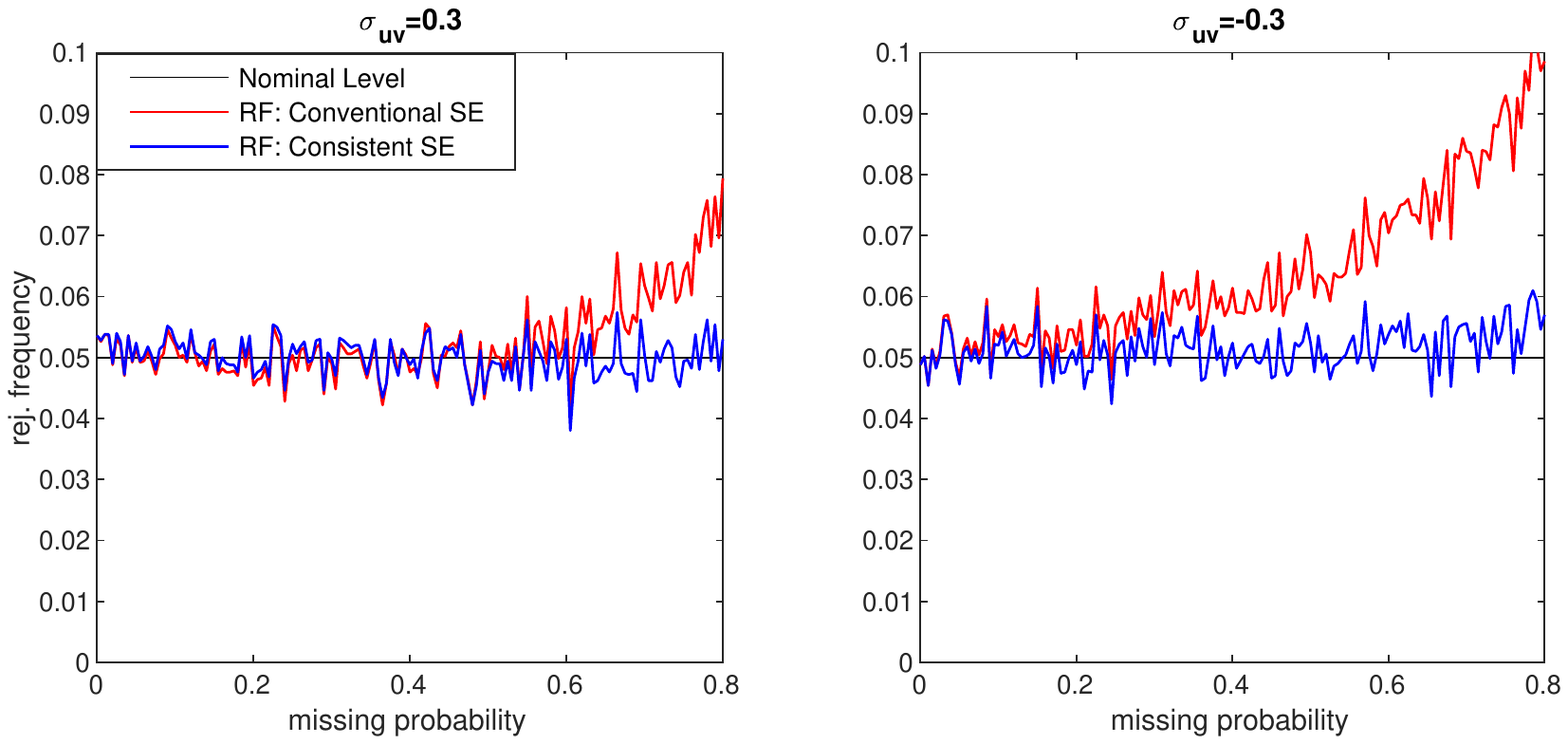}
	\caption{Consistent vs conventional variance estimator (Rejection Frequency)} \label{fig:cons_v_naive_rej}
\end{figure}

Figure \ref{fig:cons_v_naive_rmse} compares our consistent and the conventional standard errors to the observed root mean squared error (RMSE) obtained from the Monte Carlo simulations. It shows that the conventional estimator cannot properly describe the standard error of the 2SLS estimator with regression imputation. The missing probability has a linear effect on the conventional estimator (see eq. \ref{naive2SLS}) while it nonlinearly affects the true asymptotic variance as shown in Corollary 1. Moreover, as expected the conventional standard error can even decrease as the missing probability increases ($\vert\sigma_{uv}\beta\vert>\sigma^2_v \beta^2$ in the right graph of Figure \ref{fig:cons_v_naive_rmse}).   The reason why \cite{MCDONOUGH2017}, who did not derive a consistent variance estimator for regression imputation, have acceptable results in their simulations is solely due to their chosen parameters ($\sigma_{uv}>0$, $\beta>0$, low missing probability) and cannot be generalized to the entire parameter space. Our variance estimator ($\widehat{V}_{\widehat\beta_{RI}}$) performs well in both settings. Additionally, Figure \ref{fig:cons_v_naive_rej} shows how often the null hypothesis ($H_0:\,\beta=0.5$) is rejected at the $5\%$ nominal level. Again, our proposed variance estimator performs better than the conventional one---particularly, if $\sigma_{uv}<0$.

%	Conclusion
%----------------------------------------------------------------------------------------

\section{Conclusion}\label{sec:conclusion}

We investigate how two issues, which are likely present in many empirical studies, affect the estimation of causal effects, namely the endogeneity of regressors and missing data. If researchers use an instrumental variables regression after single imputation of missing values for an endogenous regressor, they have to be aware that conventional methods to estimate the variance fail to account for the imputation. The asymptotic variance of 2SLS is affected by the imputation implying that conventional methods cannot be used to construct standard errors, confidence intervals and conduct tests. We derive a heteroskedastic variance estimator which takes the imputation into account and is consistent. Monte Carlo simulations show that our estimator performs well while the conventional variance estimator does not.

%----------------------------------------------------------------------------------------
%	BIBLIOGRAPHY
%----------------------------------------------------------------------------------------

\bigskip
\bibliographystyle{jaestyle}
\bibliography{FK2019_bib.bib}

%----------------------------------------------------------------------------------------
%	APPENDICES
%----------------------------------------------------------------------------------------
\bigskip
\renewcommand{\thesection}{A}
\pagebreak

\section{Appendix}

%	A1: Asymptotic Variance
%----------------------------------------------------------------------------------------

\subsection{Proof of Proposition \ref{prop:asym_var_het}}
\label{app:asym_var_het}

Starting with 
\begin{align}
	\begin{split}
		\sqrt{n}(\widehat{\beta}_{RI}-\beta)=&\left(\left(\frac{\widetilde{x}'Z}{n}\right)\left(\frac{Z'Z}{n}\right)^{-1}\left(\frac{Z'\widetilde{x}}{n}\right)\right)^{-1}\\
		&\left(\left(\frac{\widetilde{x}'Z}{n}\right)\left(\frac{Z'Z}{n}\right)^{-1}\left(\frac{Z'\widetilde{u}}{\sqrt{n}}\right)\right) \, .
	\end{split}
\end{align}
From 2SLS \ref{ass:2sls2} and 2SLS \ref{ass:2sls3}, we have
$\frac{\widetilde{x}'P_Z\widetilde{x}}{n}\xrightarrow[]{p}{}Q_{xZ}Q_{ZZ}^{-1}Q_{Zx}$, $\frac{\widetilde{x}'Z}{n}\xrightarrow[]{p}{}Q_{xZ}$ and $\frac{Z'Z}{n}\xrightarrow[]{p}{}Q_{ZZ}$. The same holds for the complete and incomplete part of the sample due to 2SLS \ref{ass:2sls5b}, e.g., $\frac{1}{n_0}\sum_{i=1}^{n_0}Z_{i.}Z_{i.}'=Z_0'Z_0/n_0\xrightarrow[]{p}{}Q_{Z_0Z_0}$. Hence, it remains to show how $\frac{Z'\widetilde{u}}{\sqrt{n}}$ behaves as $n\xrightarrow[]{p}\infty$.
\begin{align}
	\begin{split}
		\frac{Z'\widetilde{u}}{\sqrt{n}}=&\frac{Z_0'u_0}{\sqrt{n}}+\frac{Z_1'\widetilde{u}_1}{\sqrt{n}}\\
		=&\frac{Z_0'u_0}{\sqrt{n}}+\frac{Z_1'u_1}{\sqrt{n}}+\frac{Z_1'v_1}{\sqrt{n}}\beta-\frac{n_1}{n_0}\frac{Z_1'Z_1}{n_1}\left(\frac{Z_0'Z_0}{n_0}\right)^{-1}\frac{Z_0'v_0}{\sqrt{n}}\beta\\
		=&\frac{\sqrt{n_0}}{\sqrt{n}}\left(\frac{Z_0'u_0}{\sqrt{n_0}}-\frac{n_1}{n_0}\frac{Z_1'Z_1}{n_1}\left(\frac{Z_0'Z_0}{n_0}\right)^{-1}\frac{Z_0'v_0}{\sqrt{n_0}}\beta\right)+\frac{\sqrt{n_1}}{\sqrt{n}}\left(\frac{Z_1'u_1}{\sqrt{n_1}}+\frac{Z_1'v_1}{\sqrt{n_1}}\beta\right)\\
		=&\frac{\sqrt{n_0}}{\sqrt{n}}\left(\frac{1}{\sqrt{n_0}}\sum_{i=1}^{n_0}\left(Z_{0i.}u_{0i}-\frac{n_1}{n_0}\frac{Z_1'Z_1}{n_1}\left(\frac{Z_0'Z_0}{n_0}\right)^{-1}Z_{0i.}v_{0i}\beta\right)\right)\\
		&+\frac{\sqrt{n_1}}{\sqrt{n}}\left(\frac{1}{\sqrt{n_1}}\sum_{i=n_0+1}^{n}Z_{1i.}(u_{1i}+v_{1i}\beta)\right)
	\end{split}
\end{align}
Both terms are asymptotically independent due to 2SLS \ref{ass:2sls4}, and we get from the CLT
$$\frac{\sqrt{n_0}}{\sqrt{n}}\left(\frac{1}{\sqrt{n_0}}\sum_{i=1}^{n_0}\left(Z_{0i.}u_{0i}-\frac{n_1}{n_0}\frac{Z_1'Z_1}{n_1}\left(\frac{Z_0'Z_0}{n_0}\right)^{-1}Z_{0i.}v_{0i}\beta\right)\right)\xrightarrow[]{d}{}N(0_{L},\Omega_0)$$
and
$$\frac{\sqrt{n_1}}{\sqrt{n}}\left(\frac{1}{\sqrt{n_1}}\sum_{i=n_0+1}^{n}Z_{1i.}(u_{1i}+v_{1i}\beta)\right)\xrightarrow[]{d}{}N(0_{L},\Omega_1)$$
The asymptotic variances $\Omega_0$ and $\Omega_1$ are given by,
\begin{align}
	\begin{split}
		\nonumber
		\Omega_0=&(1-p)E[Z_{0i.}Z_{0i.}'u_{0i}^2]-2pE[u_{0i}v_{0i}Z_{0i.}Z_{0i.}']Q_{Z_0Z_0}^{-1}Q_{Z_1Z_1}\beta\\
		&+\frac{p^2}{1-p}Q_{Z_1Z_1}Q_{Z_0Z_0}^{-1}E[v_{0i}^{2}Z_{0i.}Z_{0i.}']Q_{Z_0Z_0}^{-1}Q_{Z_1Z_1}\beta^2
	\end{split}
\end{align}
and,
\begin{align}
	\begin{split}
		\nonumber
		\Omega_1=&p\left(E[Z_{1i.}Z_{1i.}'u_{1i}^2]+2E[u_{1i}v_{1i}Z_{1i.}Z_{1i.}']\beta+E[Z_{1i.}Z_{1i.}'v_{1i}^2]\beta^2\right)
	\end{split}
\end{align}
Combing both gives the limiting distribution of $\frac{Z'\widetilde{u}}{\sqrt{n}}$,
\begin{align}
	\begin{split}
		\frac{Z'\widetilde{u}}{\sqrt{n}}&\xrightarrow[]{p}{}N(0_{L},\Omega)\\
		where\;\Omega=&E[Z_{i.}Z_{i.}'u_{i}^2]-2pE[u_{0i}v_{0i}Z_{0i.}Z_{0i.}']Q_{Z_0Z_0}^{-1}Q_{Z_1Z_1}\beta\\
		&+\frac{p^2}{1-p}Q_{Z_1Z_1}Q_{Z_0Z_0}^{-1}E[v_{0i}^{2}Z_{0i.}Z_{0i.}']Q_{Z_0Z_0}^{-1}Q_{Z_1Z_1}\beta^2\\
		&+2pE[u_{1i}v_{1i}Z_{1i.}Z_{1i.}']\beta+pE[Z_{1i.}Z_{1i.}'v_{1i}^2]\beta^2 \, ,
	\end{split}
\end{align}
where we used $(1-p)E[Z_{0i.}Z_{0i.}'u_{0i}^2]+pE[Z_{1i.}Z_{1i.}'u_{1i}^2]=E[Z_{i.}Z_{i.}'u_{i}^2]$. Using Slutsky's lemma, we can then prove Proposition \ref{prop:asym_var_het}
\begin{align}
	\begin{split}
		\sqrt{n}(\widehat{\beta}_{RI}-\beta)&\xrightarrow[]{d}\left(Q_{xZ}Q_{ZZ}^{-1}Q_{Zx}\right)^{-1}Q_{xZ}Q_{ZZ}^{-1}N(0_{L},\Omega)=N(0,V_{\widehat\beta_{RI}})\\
		where\;&V_{\widehat\beta_{RI}}=\left(Q_{xZ}Q_{ZZ}^{-1}Q_{Zx}\right)^{-1}Q_{xZ}Q_{ZZ}^{-1}\,\Omega{}\,Q_{ZZ}^{-1}Q_{Zx}\left(Q_{xZ}Q_{ZZ}^{-1}Q_{Zx}\right)^{-1} \, .
	\end{split}
\end{align}

%	A2: Heteroskedasticity robust variance estimator
%----------------------------------------------------------------------------------------

\subsection{Proof of Proposition \ref{prop:consistent_rob_var}}
\label{app:consistent_rob_var}
Starting with
\begin{align}
	\begin{split}
		n\widehat{V}_{\widehat{\beta}_{RI}}=&\left(\frac{\widetilde{x}'P_Z\widetilde{x}}{n}\right)^{-1}\frac{\widetilde{x}'Z}{n}\left(\frac{Z'Z}{n}\right)^{-1}\frac{\widehat{W}_{RI}}{n}\left(\frac{Z'Z}{n}\right)^{-1}\frac{Z'\widetilde{x}}{n}\left(\frac{\widetilde{x}'P_Z\widetilde{x}}{n}\right)^{-1} \, .
	\end{split}
\end{align}
From 2SLS \ref{ass:2sls2} and 2SLS \ref{ass:2sls3}, we have
$\frac{\widetilde{x}'P_Z\widetilde{x}}{n}\xrightarrow[]{p}{}Q_{xZ}Q_{ZZ}^{-1}Q_{Zx}$, $\frac{\widetilde{x}'Z}{n}\xrightarrow[]{p}{}Q_{xZ}$ and $\frac{Z'Z}{n}\xrightarrow[]{p}{}Q_{ZZ}$. The same holds for the complete and incomplete part of the sample due to 2SLS \ref{ass:2sls5b}, e.g., $\frac{1}{n_0}\sum_{i=1}^{n_0}Z_{i.}Z_{i.}'=Z_0'Z_0/n_0\xrightarrow[]{p}{}Q_{Z_0Z_0}$. Hence, it remains to show that $\frac{\widehat{W}_{RI}}{n}=\widehat{\Omega}\xrightarrow[]{p}\Omega$ as $n\xrightarrow[]{p}\infty$.
\begin{align} \label{eq:rob_var_est}
	\small
	\begin{split}
		\widehat{\Omega}=&\left(\frac{1}{n}\sum_{i=1}^{n}\widehat{\widetilde{u}}_{i}^2Z_{i.}Z_{i.}'\right)-2\frac{n_1}{n}\left(\frac{1}{n_0}\sum_{i=1}^{n_0}\widehat{\widetilde{u}}_{i}\widehat{\widetilde{v}}_{i}Z_{i.}Z_{i.}'\right)\left(\frac{1}{n_0}\sum_{i=1}^{n_0}Z_{i.}Z_{i.}'\right)^{-1}\left(\frac{1}{n_1}\sum_{i=n_0+1}^{n}Z_{i.}Z_{i.}'\right)\widehat{\beta}_{RI}\\
		&+\frac{n_1^2}{nn_0}\left(\frac{1}{n_1}\sum_{i=n_0+1}^{n}Z_{i.}Z_{i.}'\right)\left(\frac{1}{n_0}\sum_{i=1}^{n_0}Z_{i.}Z_{i.}'\right)^{-1}\left(\frac{1}{n_0}\sum_{i=1}^{n_0}\widehat{\widetilde{v}}_{i}^2Z_{i.}Z_{i.}'\right)\left(\frac{1}{n_0}\sum_{i=1}^{n_0}Z_{i.}Z_{i.}'\right)^{-1}\left(\frac{1}{n_1}\sum_{i=n_0+1}^{n}Z_{i.}Z_{i.}'\right)\widehat{\beta}_{RI}^2\\
		&-\frac{n_1}{n}\left(\frac{1}{n_0}\right)\left(\frac{1}{n_1}\sum_{i=n_0+1}^{n}Z_{i.}Z_{i.}'Z_{i.}Z_{i.}'\right)\left(\frac{1}{n_0}\sum_{i=1}^{n_0}Z_{i.}Z_{i.}'\right)^{-1}\left(\frac{1}{n_0}\sum_{i=1}^{n_0}\widehat{\widetilde{v}}_{i}^2Z_{i.}Z_{i.}'\right)\left(\frac{1}{n_0}\sum_{i=1}^{n_0}Z_{i.}Z_{i.}'\right)^{-1}\widehat{\beta}_{RI}^2\\
		where&\;\widehat{\widetilde{u}}_i=y_i-\widetilde{x}_i\widehat\beta_{RI}\;and\; \widehat{\widetilde{v}}_i=\widetilde{x}_i-Z_{i.}'\widehat\pi_{CC}  \, .
	\end{split}
\end{align}
In the next steps, it is useful to rewrite the composite residuals as $\widehat{\widetilde{u}}_{i}=\widetilde{u}_i-\widetilde{x}_i(\widehat{\beta}_{RI}-\beta)$ and $\widehat{\widetilde{v}}_{i}=\widetilde{v}_i-Z_{i.}'(\widehat{\pi}_{CC}-\pi)$. Rewritting the first term of eq. (\ref{eq:rob_var_est}), we get
\begin{align}\label{eq:rob_var_res}
	\small
	\begin{split}
		\frac{1}{n}\sum_{i=1}^{n}\widehat{\widetilde{u}}_{i}^2Z_{i.}Z_{i.}'=&\left(\frac{1}{n}\sum_{i=1}^{n}(\widetilde{u}_i-\widetilde{x}_i(\widehat{\beta}_{RI}-\beta))^2Z_{i.}Z_{i.}'\right)\\
		=&\left(\frac{1}{n}\sum_{i=1}^{n}\widetilde{u}_i^2Z_{i.}Z_{i.}'\right)-2\left(\frac{1}{n}\sum_{i=1}^{n}\widetilde{u}_i\widetilde{x}_i(\widehat{\beta}_{RI}-\beta)Z_{i.}Z_{i.}'\right)+\left(\frac{1}{n}\sum_{i=1}^{n}\widetilde{x}_i^2(\widehat{\beta}_{RI}-\beta)^2Z_{i.}Z_{i.}'\right)\\
		=&\left(\frac{1}{n}\sum_{i=1}^{n}\widetilde{u}_i^2Z_{i.}Z_{i.}'\right)-2(\widehat{\beta}_{RI}-\beta)\left(\frac{1}{n}\sum_{i=1}^{n}\widetilde{u}_i\widetilde{x}_iZ_{i.}Z_{i.}'\right)+(\widehat{\beta}_{RI}-\beta)^2\left(\frac{1}{n}\sum_{i=1}^{n}\widetilde{x}_i^2Z_{i.}Z_{i.}'\right)
	\end{split}
\end{align}
For the first term of eq. (\ref{eq:rob_var_res}), we have
\begin{align}\label{first_term}
	\small
	\begin{split}
		\frac{1}{n}\sum_{i=1}^{n}\widetilde{u}_i^2Z_{i.}Z_{i.}'=&\frac{n_0}{n}\left(\frac{1}{n_0}\sum_{i=1}^{n_0}u_i^2Z_{i.}Z_{i.}'\right)+\frac{n_1}{n}\left(\frac{1}{n_1}\sum_{i=n_0+1}^{n}\widetilde{u}_i^2Z_{i.}Z_{i.}'\right)\\
		=&\frac{n_0}{n}\left(\frac{1}{n_0}\sum_{i=1}^{n_0}u_i^2Z_{i.}Z_{i.}'\right)+\frac{n_1}{n}\left(\frac{1}{n_1}\sum_{i=n_0+1}^{n}u_i^2Z_{i.}Z_{i.}'\right)\\
		&+2\frac{n_1}{n}\left(\frac{1}{n_1}\sum_{i=n_0+1}^{n}v_iu_iZ_{i.}Z_{i.}'\right)\beta+\left(\frac{1}{n_1}\sum_{i=n_0+1}^{n}v_i^2Z_{i.}Z_{i.}'\right)\beta^2\\
		&-2\frac{n_1}{n}\left(\frac{1}{n_1}\sum_{i=n_0+1}^{n}\widetilde{u}_iZ_{i.}'(Z_0'Z_0)^{-1}Z_0'v_0Z_{i.}Z_{i.}'\right)\beta\\
		&-\frac{n_1}{n}\left(\frac{1}{n_1}\sum_{i=n_0+1}^{n}Z_{i.}Z_{i.}'(Z_0'Z_0)^{-1}Z_0'v_0v_0'Z_0(Z_0'Z_0)^{-1}Z_{i.}Z_{i.}'\right)\beta^2\\
		\xrightarrow[]{p}&E[u_{i}^2Z_{i.}Z_{i.}']+2pE[u_{1i}v_{1i}Z_{1i.}Z_{1i.}']\beta+pE[v_{1i}^2Z_{1i.}Z_{1i.}']\beta^2
	\end{split}
\end{align}
Using $(Z_0'Z_0)^{-1}Z_0'v_0=\widehat{\pi}_{CC}-\pi\, $, the triangle inequality, Hölder's inequality, we get for the next to last term of eq. (\ref{first_term})
\begin{align}
	\begin{split}
		\nonumber
		\left\Vert\frac{1}{n_1}\sum_{i=n_0+1}^{n}Z_{i.}Z_{i.}'\widetilde{u}_iZ_{i.}'(\widehat{\pi}_{CC}-\pi)\right\Vert&\leq\frac{1}{n_1}\sum_{i=n_0+1}^{n}\left\Vert{}Z_{i.}Z_{i.}'\widetilde{u}_iZ_{i.}'(\widehat{\pi}_{CC}-\pi)\right\Vert\\
		&\leq\frac{1}{n_1}\sum_{i=n_0+1}^{n}\Vert{}Z_{i.}\Vert{}^3|\widetilde{u}_i|\Vert\widehat{\pi}_{CC}-\pi\Vert=o_p(1) \, ,
	\end{split}
\end{align}
as $\Vert\widehat{\pi}_{CC}-\pi\Vert\xrightarrow[]{p}0$ and $E[\Vert{}Z_{i.}\Vert{}^3|\widetilde{u}_i|]\leq{}E[\Vert{}Z_{i.}\Vert{}^4]^{\frac{3}{4}}E[|\widetilde{u}_i|^4]^{\frac{1}{4}}<\infty$ from $E[\Vert{}Z_{i.}\Vert{}^4]$, $E[y_i^4]$ and $E[\widetilde{x}_i^4]$ being finite. Similarly, we have for the last term of eq. (\ref{first_term})
\begin{align}
	\begin{split}
		\nonumber
		\left\Vert\frac{1}{n_1}\sum_{i=n_0+1}^{n}Z_{i.}Z_{i.}'Z_{i.}Z_{i.}'(\widehat{\pi}_{CC}-\pi)(\widehat{\pi}_{CC}-\pi)'\right\Vert&\leq\frac{1}{n_1}\sum_{i=n_0+1}^{n}\left\Vert{}Z_{i.}Z_{i.}'Z_{i.}Z_{i.}'(\widehat{\pi}_{CC}-\pi)(\widehat{\pi}_{CC}-\pi)'\right\Vert\\
		&\leq\frac{1}{n_1}\sum_{i=n_0+1}^{n}\Vert{}Z_{i.}\Vert{}^4\Vert\widehat{\pi}_{CC}-\pi\Vert^2=o_p(1) \, .
	\end{split}
\end{align}

\noindent
Now, we show that the last two terms of equation (\ref{eq:rob_var_res}) converge towards zero as $n\xrightarrow[]{p}\infty$. 
\begin{align}
	\small
	\begin{split}
		\nonumber
		\left\Vert(\widehat{\beta}_{RI}-\beta)\frac{1}{n}\sum_{i=1}^{n}\widetilde{u}_i\widetilde{x}_iZ_{i.}Z_{i.}'\right\Vert&\leq|\widehat{\beta}_{RI}-\beta|\frac{1}{n_1}\sum_{i=n_0+1}^{n}\left\Vert{}\widetilde{u}_i\widetilde{x}_iZ_{i.}Z_{i.}'\right\Vert\\
		&\leq|\widehat{\beta}_{RI}-\beta|\frac{1}{n_1}\sum_{i=n_0+1}^{n}|\widetilde{u}_i||\widetilde{x}_i|\Vert{}Z_{i.}\Vert{}^2=o_p(1) \, ,
	\end{split}
\end{align}
and
\begin{align}
	\small
	\begin{split}
		\nonumber
		\left\Vert(\widehat{\beta}_{RI}-\beta)^{2}\frac{1}{n}\sum_{i=1}^{n}\widetilde{x}_i^2Z_{i.}Z_{i.}'\right\Vert&\leq|\widehat{\beta}_{RI}-\beta|^2\frac{1}{n_1}\sum_{i=n_0+1}^{n}\left\Vert{}\widetilde{x}_i^2Z_{i.}Z_{i.}'\right\Vert\\
		&\leq|\widehat{\beta}_{RI}-\beta|^2\frac{1}{n_1}\sum_{i=n_0+1}^{n}|\widetilde{x}_i|^2\Vert{}Z_{i.}\Vert{}^2=o_p(1) \, ,
	\end{split}
\end{align}
since $|\widehat{\beta}_{RI}-\beta|\xrightarrow[]{p}0 \, $, $E[|\widetilde{u}_i||\widetilde{x}_i|\Vert{}Z_{i.}\Vert{}^2]\leq{}E[|\widetilde{u}_i|^{4}]^{\frac{1}{4}}E[|\widetilde{x}_i|^{4}]^{\frac{1}{4}}E[\Vert{}Z_{i.}\Vert{}^4]^{\frac{2}{4}}<\infty$, and $E[|\widetilde{x}_i|^2\Vert{}Z_{i.}\Vert{}^2]\leq{}E[|\widetilde{x}_i|^{4}]^{\frac{1}{2}}E[\Vert{}Z_{i.}\Vert{}^4]^{\frac{1}{2}}<\infty$. 

\noindent
This proves that,
\begin{align}
	\small
	\begin{split}
		\frac{1}{n}\sum_{i=1}^{n}\widehat{\widetilde{u}}_{i}^2Z_{i.}Z_{i.}'\xrightarrow[]{p}{}E[u_{i}^2Z_{i.}Z_{i.}']+2pE[u_{1i}v_{1i}Z_{1i.}Z_{1i.}']\beta+pE[v_{1i}^2Z_{1i.}Z_{1i.}']\beta^2
	\end{split}
\end{align}
Now, analyzing the remaining terms of eq. (\ref{eq:rob_var_est}), we get
\begin{align}
	\small
	\begin{split}
		\frac{1}{n_0}\sum_{i=1}^{n_0}\widehat{\widetilde{u}}_{i}\widehat{\widetilde{v}}_{i}Z_{i.}Z_{i.}'=&\frac{1}{n_0}\sum_{i=1}^{n_0}(\widetilde{u}_i-\widetilde{x}_i(\widehat{\beta}_{RI}-\beta))(\widetilde{v}_i-Z_{i.}'(\widehat{\pi}_{CC}-\pi))Z_{i.}Z_{i.}'\\
		=&\frac{1}{n_0}\sum_{i=1}^{n_0}\widetilde{u}_{i}\widetilde{v}_{i}Z_{i.}Z_{i.}'+\frac{1}{n_0}\sum_{i=1}^{n_0}\widetilde{x}_i(\widehat{\beta}_{RI}-\beta)\widetilde{v}_{i}Z_{i.}Z_{i.}'-\frac{1}{n_0}\sum_{i=1}^{n_0}\widetilde{u}_{i}Z_{i.}'(\widehat{\pi}_{CC}-\pi)Z_{i.}Z_{i.}'\\
		&+\frac{1}{n_0}\sum_{i=1}^{n_0}\widetilde{x}_i(\widehat{\beta}_{RI}-\beta)Z_{i.}'(\widehat{\pi}_{CC}-\pi)Z_{i.}Z_{i.}'\\
		\xrightarrow[]{p}&E[u_{0i}v_{0i}Z_{0i.}Z_{0i.}'] \, ,
	\end{split}
\end{align}
and 
\begin{align}
	\small
	\begin{split}
		\frac{1}{n_0}\sum_{i=1}^{n_0}\widehat{\widetilde{v}}_{i}^2Z_{i.}Z_{i.}'=&\frac{1}{n_0}\sum_{i=1}^{n_0}\widetilde{v}_i-Z_{i.}'(\widehat{\pi}_{CC}-\pi))^2Z_{i.}Z_{i.}'\\
		=&\frac{1}{n_0}\sum_{i=1}^{n_0}\widetilde{v}_i^2Z_{i.}Z_{i.}'-2\frac{1}{n_0}\sum_{i=1}^{n_0}(\widetilde{v}_iZ_{i.}'(\widehat{\pi}_{CC}-\pi)Z_{i.}Z_{i.}'+\frac{1}{n_0}\sum_{i=1}^{n_0}(\widehat{\pi}_{CC}-\pi)'Z_{i.}Z_{i.}'(\widehat{\pi}_{CC}-\pi)Z_{i.}Z_{i.}'\\
		\xrightarrow[]{p}{}&E[v_{0i}^2Z_{0i.}Z_{0i.}'].
	\end{split}
\end{align}
The proofs are similar to the results above and omitted here.

\noindent
Finally, we can show that
\begin{align}
	\small
	\begin{split}
		\widehat{\Omega}=&\left(\frac{1}{n}\sum_{i=1}^{n}\widehat{\widetilde{u}}_{i}^2Z_{i.}Z_{i.}'\right)-2\frac{n_1}{n}\left(\frac{1}{n_0}\sum_{i=1}^{n_0}\widehat{\widetilde{u}}_{i}\widehat{\widetilde{v}}_{i}Z_{i.}Z_{i.}'\right)\left(\frac{1}{n_0}\sum_{i=1}^{n_0}Z_{i.}Z_{i.}'\right)^{-1}\left(\frac{1}{n_1}\sum_{i=n_0+1}^{n}Z_{i.}Z_{i.}'\right)\widehat{\beta}_{RI}\\
		&+\frac{n_1^2}{nn_0}\left(\frac{1}{n_1}\sum_{i=n_0+1}^{n}Z_{i.}Z_{i.}'\right)\left(\frac{1}{n_0}\sum_{i=1}^{n_0}Z_{i.}Z_{i.}'\right)^{-1}\left(\frac{1}{n_0}\sum_{i=1}^{n_0}\widehat{\widetilde{v}}_{i}^2Z_{i.}Z_{i.}'\right)\left(\frac{1}{n_0}\sum_{i=1}^{n_0}Z_{i.}Z_{i.}'\right)^{-1}\left(\frac{1}{n_1}\sum_{i=n_0+1}^{n}Z_{i.}Z_{i.}'\right)\widehat{\beta}_{RI}^2\\
		&-\frac{n_1}{n}\left(\frac{1}{n_0}\right)\left(\frac{1}{n_1}\sum_{i=n_0+1}^{n}Z_{i.}Z_{i.}'Z_{i.}Z_{i.}'\right)\left(\frac{1}{n_0}\sum_{i=1}^{n_0}Z_{i.}Z_{i.}'\right)^{-1}\left(\frac{1}{n_0}\sum_{i=1}^{n_0}\widehat{\widetilde{v}}_{i}^2Z_{i.}Z_{i.}'\right)\left(\frac{1}{n_0}\sum_{i=1}^{n_0}Z_{i.}Z_{i.}'\right)^{-1}\widehat{\beta}_{RI}^2\\
		\xrightarrow[]{p}&E[Z_{i.}Z_{i.}'u_{i}^2]-2pE[u_{0i}v_{0i}Z_{0i.}Z_{0i.}']Q_{Z_0Z_0}^{-1}Q_{Z_1Z_1}\beta\\
		&+\frac{p^2}{1-p}Q_{Z_1Z_1}Q_{Z_0Z_0}^{-1}E[v_{0i}^{2}Z_{0i.}Z_{0i.}']Q_{Z_0Z_0}^{-1}Q_{Z_1Z_1}\beta^2\\
		&+2pE[u_{1i}v_{1i}Z_{1i.}Z_{1i.}']\beta+pE[Z_{1i.}Z_{1i.}'v_{1i}^2]\beta^2=\Omega \, .
	\end{split}
\end{align}
The last term of the estimator converges toward zero due to the $\frac{1}{n_0}$ term. Nevertheless, we keep it in the estimator to improve its finite sample performance as this term is essentially part of the first and third term and hence should be deducted once. The limit is not affected by this choice.

\bigskip

%	A3: Homoskedastic Variance
%----------------------------------------------------------------------------------------

\subsection{Proof of Corollary 1}
\label{app:hom_var}

Using the results from Proposition \ref{prop:asym_var_het} and the simplifications induced by homoskedasticity, e.g., $E[u_i^2Z_{i.}Z_{i.}']=Q_{ZZ}\sigma_{u}^2$, we get
\begin{align}
	\small
	\begin{split}
		\Omega^{hom}=&Q_{ZZ}\sigma_{u}^2-2pQ_{Z_1Z_1}\sigma_{uv}^2\beta+2pQ_{Z_1Z_1}\sigma_{uv}^2\beta+pQ_{Z_1Z_1}\sigma_{v}^2\beta^2+p\frac{p}{1-p}Q_{Z_1Z_1}Q_{Z_0Z_0}^{-1}Q_{Z_1Z_1}\sigma_{v}^2\beta^2\\
		=&Q_{ZZ}\sigma_{u}^2+pQ_{Z_1Z_1}\sigma_{v}^2\beta^2+p\frac{p}{1-p}Q_{Z_1Z_1}Q_{Z_0Z_0}^{-1}Q_{Z_1Z_1}\sigma_{v}^2\beta^2\\
		=&Q_{ZZ}\sigma_{u}^2-Q_{ZZ}\sigma_{v}^2\beta^2+\frac{1}{1-p}Q_{ZZ}Q_{Z_0Z_0}^{-1}Q_{ZZ}\sigma_{v}^2\beta^2 \, ,
	\end{split}
\end{align}
where we use $Q_{ZZ}=pQ_{Z_1Z_1}+(1-p)Q_{Z_0Z_0}$ in the last equality. The asymptotic variance is
\begin{align}
	\begin{split}
		V_{\widehat\beta_{RI}}^{hom}=&\left(Q_{xZ}Q_{ZZ}^{-1}Q_{Zx}\right)^{-1}\sigma_{u}^2-\left(Q_{xZ}Q_{ZZ}^{-1}Q_{Zx}\right)^{-1}\sigma_{v}^2\beta^2\\
		&+\frac{1}{1-p}\left(Q_{xZ}Q_{ZZ}^{-1}Q_{Zx}\right)^{-1}\left(Q_{xZ}Q_{Z_0Z_0}^{-1}Q_{Zx}{}\right)\\
		&\phantom{+}\left(Q_{xZ}Q_{ZZ}^{-1}Q_{Zx}\right)^{-1}\sigma_{v}^2\beta^2 \, .
	\end{split} \nonumber
\end{align}
Under MCAR, we know $Q_{Z_0Z_0}=Q_{ZZ}$, which we can use to derive the asymptotic variance of $\widehat\beta_{RI}$ under homoskedasticity and MCAR.

\end{document}